\makeatletter
\newcommand{\dontusepackage}[2][]{%
  \@namedef{ver@#2.sty}{9999/12/31}%
  \@namedef{opt@#2.sty}{#1}}
\makeatother

\dontusepackage[sort&compress]{natbib}
\newcommand{\bibpunct}[6]{}

\documentclass[oneside,nocrop]{ipart_v1}

\usepackage{t1enc}
\usepackage[latin1]{inputenc}
\usepackage[english]{babel}
\usepackage[normalem]{ulem}

\usepackage{amsthm}
\usepackage{yfonts}

\usepackage{bbm}
\usepackage{bm}
\usepackage{mathrsfs}

\usepackage{amsrefs}
\usepackage{enumerate}
\usepackage{tabu}
\usepackage{tikz-cd}
\usepackage{upgreek}

\numberwithin{equation}{section}

\theoremstyle{plain}

\newtheorem{corollary}[subsection]{Corollary}
\newtheorem{lemma}[subsection]{Lemma}
\newtheorem{proposition}[subsection]{Proposition}
\newtheorem{theorem}[subsection]{Theorem}

\theoremstyle{definition}

\newtheorem{definition}[subsection]{Definition}

\newcommand{\CA}{\mathcal{A}}
\newcommand{\Ahat}{\widehat{\mathcal{A}}}
\newcommand{\tCA}{\widetilde{\mathcal{A}}}
\renewcommand{\AA}{\mathbb{A}}
\newcommand{\AAhat}{\widehat{\mathbb{A}}}
\newcommand{\tAA}{\widetilde{\mathbb{A}}}
\newcommand{\FF}{\mathbb{F}}
\newcommand{\CC}{\mathcal{C}}
\newcommand{\CF}{\mathcal{F}}
\newcommand{\CO}{\mathcal{O}}
\newcommand{\CU}{\mathcal{U}}
\newcommand{\CV}{\mathcal{V}}
\newcommand{\CW}{\mathcal{W}}

\newcommand{\Om}{\Omega}
\newcommand{\Cech}{{\v{C}ech}}
\newcommand{\Ss}{\mathbb{S}}
\newcommand{\half}{\tfrac12}
\newcommand{\eps}{\epsilon}
\let\Re=\undef
\DeclareMathOperator{\Re}{Re}
\newcommand{\Wedge}{\Lambda}

\newcommand{\R}{\mathbb{R}}
\newcommand{\Z}{\mathbb{Z}}
\newcommand{\N}{\mathbb{N}}
\renewcommand{\H}{\mathbb{H}}
\newcommand{\T}{\mathsf{T}}
\newcommand{\V}{V}
\newcommand{\p}{\partial}
\newcommand{\<}{\langle}
\renewcommand{\>}{\rangle}
\newcommand{\tint}{{\textstyle\int}}

\DeclareMathOperator{\Spin}{Spin}
\DeclareMathOperator{\SU}{SU}
\DeclareMathOperator{\SO}{SO}
\DeclareMathOperator{\so}{\mathfrak{so}}
\DeclareMathOperator{\ad}{ad}
\DeclareMathOperator{\im}{im}
\DeclareMathOperator{\gh}{gh}
\DeclareMathOperator{\pr}{pr}

\renewcommand{\S}{\mathsf{S}}
\newcommand{\G}{\mathsf{G}}
\newcommand{\s}{\mathsf{s}}
\renewcommand{\q}{\mathsf{q}}
\newcommand{\D}{\mathsf{D}}
\DeclareMathOperator{\pa}{\mathit{p}}

\renewcommand{\phi}{\varphi}
\newcommand{\bull}{\bullet}
\DeclareMathOperator{\gr}{gr}

\newcommand{\Cone}{C}
\newcommand{\Phase}{M}
\renewcommand{\[}{(\!(}
\renewcommand{\]}{)\!)}

\begin{document}

\title{Covariance of the classical Brink--Schwarz superparticle}

\author{Ezra Getzler and Sean Weinz Pohorence}

\begin{abstract}
  We show that the classical Brink--Schwarz superparticle is a
  generalized AKSZ field theory. We work in the Batalin--Vilkovisky
  formalism: the main technical tool is the vanishing of
  Batalin--Vilkovisky cohomology below degree $-1$.
\end{abstract}

\maketitle

\thispagestyle{empty}

\section{Introduction}

The construction of string theory is based on a nonlinear sigma-model
with worldsheet and target a Riemann surface $\Sigma$ and
$26$-dimensional Minkowski space respectively. The superstring, a
supersymmetric analogue of the string, may be represented in several
ways. In the Neveu--Schwarz/Ramond superstring, $\Sigma$ is a Riemann
supersurface, the target is $10$-dimensional Minkowski space, and the
target-space supersymmetry of the theory is not apparent at the
classical level. By contrast, in the Green--Schwarz superstring,
$\Sigma$ is a Riemann surface, while the target is a superspace whose
underlying vector space is the same ten-dimensional Minkowski space,
but which incorporates an additional fermionic chiral Majorana-Weyl
spinor. In this model, the target-space supersymmetry is manifest, at
the cost of quantization being considerably more complicated. (There
are other approaches to the superstring, such as the pure spinor
theory of Berkovits \cite{Berkovits}.)

There has been a lot of interest in the study of toy models of these
two formulations of superstring theory, in which the Riemann surface
$\Sigma$ is replaced by a one-dimensional manifold, the worldline: the
resulting toy models of the Neveu--Schwarz/Ramond superstring and
Green--Schwarz superstring are respectively known as the spinning
particle \cite{BDZDH} and the superparticle \cite{BS}. In this paper,
we focus on the superparticle.

The Batalin--Vilkovisky formalism for classical field theories
provides a powerful way of encoding symmetries: it is especially
adapted to extending on-shell symmetries, such as extended
supersymmetry and supergravity, off-shell. Poisson geometry studies
Maurer--Cartan elements in the Schouten algebra of a manifold: the
classical Batalin--Vilkovisky master equation generalizes this to
graded supermanifolds, and then further generalizes from finite
dimensional calculus to variational calculus.

One of the principles in the construction of solutions of the
classical master equation of Batalin and Vilkovisky is that the
cohomology of the Batalin--Vilkovisky differential on local
functionals, known as the (classical) Batalin--Vilkovisky cohomology,
is bounded below. In the absence of worldsheet supersymmetry, it
appears to be a general phenomenon that these cohomology groups vanish
for ghost number less than $-d$, where $d$ is the dimension of the
worldsheet. This condition may be violated in the presence of
worldsheet supersymmetry: as shown in \cite{cohomology}, the classical
Batalin--Vilkovisky cohomology of the spinning particle is nontrivial
at arbitrarily large negative ghost number.

The superparticle does not have worldsheet supersymmetry, unlike the
spinning particle. One of the main results of this paper is that the
classical Batalin--Vilkovisky cohomology for the superparticle,
suitably interepreted, vanishes below ghost number $-1$.

Using this result, we show that the classical superparticle is a
covariant field theory, in the sense of \cite{covariant}. Let $\int\S$
be the action of the superparticle, which satisfies the classical
Batalin--Vilkovisky master equation
\begin{equation*}
  \half ( \tint \S , \tint \S ) = 0 .
\end{equation*}
We show that there is a power series
\begin{equation*}
  \S_u = \S + \sum_{n=0}^\infty u^{n+1} \, \G_n ,
\end{equation*}
where $\G_n$ is a density of ghost number $-2n-2$ and $u$ is a formal
variable of ghost number $2$, which satisfies the curved
Maurer--Cartan equation
\begin{equation}
  \label{covariant}
  \delta\S_u + \half ( \tint \S_u , \tint \S_u ) = - u \tint \D .
\end{equation}
Here, $\D$ is the density
\begin{equation*}
  \D = x^+_\mu \p x^\mu + p^{+\mu} \p p_\mu - e \p e^+ + c^+ \p c
  + \sum_{n=0}^\infty \T( \theta^+_n , \p\theta_n) ,
\end{equation*}
whose associated Batalin--Vilkovisky Hamiltonian vector field equals
$\p$, the generator of time translation along the worldline. In
particular,
\begin{equation}
  \label{g0}
  ( \delta + \s ) \bigl( \tint \G_0 \bigr) = - \tint \D .
\end{equation}
Possessing a solution to \eqref{covariant} is a property that the
superparticle shares with AKSZ theories \cite{AKSZ}: in this sense,
the superparticle may be viewed as a generalized AKSZ theory.

We solve \eqref{covariant} over the open subset of the phase space
where the momentum $p_\mu$ is nonzero. For the sheaves that we
consider, this subset has nontrivial cohomology. The usual way to deal
with this would be to work with \Cech\ cochains.  Unfortunately, there
is no way to extend the Batalin--Vilkovisky antibracket to the space
of \Cech\ cochains: this is related to the failure of the cup-product
to be graded commutative. We circumvent this difficulty by working
with Sullivan's Thom--Whitney complex \cite{Sullivan}: this replaces
simplicial cochains by differential forms on simplices, and allows the
definition of the antibracket at the cochain level. In passing, we
note that in Berkovits's description of superstrings using pure
spinors \cite{Berkovits}, a similar issue arises, involving the cone
of pure spinors with the origin removed.

The action of the superparticle $\S$ in first-order formalism is
globally defined, and was found by Lindstr\"om et al.\
\cite{LRSVV}. We construct the next term in the expansion of $\S_u$
explicitly as a solution to \eqref{g0}. For $n>0$, the coefficient
$\G_n$ of $u^{n+1}$ in $\S_u$ has ghost number $-2n-2$, and is a
solution to the equation
\begin{equation*}
  (\delta+\s) \bigl( \tint \G_n \bigr) = - \frac{1}{2} \sum_{j+k=n-1}
  ( \tint \G_j , \tint \G_k ) . 
\end{equation*}
Since the cohomology of the operator $\delta+\s$ vanishes below degree
$-1$, we may solve this equation.

Let $C^*(\so(9,1))$ be the graded commutative algebra of the Lie
algebra $\so(9,1)$ of the Lorentz group (the exterior algebra
generated by $\so(9,1)^\vee$), with differential $d$.  The action of
the Lie algebra $\so(9,1)$ on the space of fields of the superparticle
is generated by the Batalin--Vilkovisky currents
\begin{equation}
  \label{Lorentz_alg}
  M^{\mu\nu} = \eta^{\lambda[\mu} x^{\nu]} x^+_\lambda
  - \eta^{\lambda[\mu} p^{+\nu]} p_\lambda
  - \sum_{n=0}^\infty\T^{\mu\nu}(\theta^+_n,\theta_n) .
\end{equation}
Let $\S(\eps)=\S+M^{\mu\nu}\eps_{\mu\nu}$, where $\eps_{\mu\nu}$ is
the dual basis of $\so(9,1)^\vee$; this is an element of total degree
$0$ in the tensor product of $C^*(\so(9,1))$ and the
Batalin--Vilkovisky graded Lie algebra. The Lorentz invariance of the
action $\S$ may be expressed by the following extension of the
classical master equation:
\begin{equation*}
  d\tint\S(\eps) + \half ( \tint \S(\eps) , \tint \S(\eps) ) = 0 .
\end{equation*}

The Lorentz group does not act on Thom--Whitney complex, because the
open cover itself is not Lorentz invariant; in particular, $\G_0$ is
not invariant under the action of $\so(9,1)$. Nevertheless, it may be
proved that $\S_u$ has an enhancement
\begin{equation}
  \label{Lorentz}
  \S_u(\eps) = \S(\eps) + \sum_{n=0}^\infty u^{n+1} \G_n(\eps) ,
\end{equation}
where $\G_n(\eps)$ is an element of total degree $-2n-2$ in the tensor
product of $C^*(\so(9,1))$ and the Thom--Whitney extension of the
Batalin--Vilkovisky graded Lie algebra, such that the following
extension of \eqref{covariant} holds:
\begin{equation}
  \label{covariant:Lorentz}
  (d+\delta) \tint \S_u(\eps) + \half ( \tint \S_u(\eps) , \tint
  \S_u(\eps) ) = - u \tint \D .
\end{equation}
In mathematical terms, this equation, which is nothing but the BRST
formalism for the global symmetry Lie algebra $\so(9,1)$, expresses
that the covariant field theory is Lorentz invariant up to homotopy.

It may be verified that $\S$ and $\G_0$ are invariant under
supersymmetry. We also show that the terms $\G_n(\eps)$, $n\ge0$, may
be chosen to be invariant under supersymmetry.

\section{The classical Batalin--Vilkovisky master equation}

Consider a Batalin--Vilkovisky model with fields $\{\xi^a\}_{a\in I}$,
of ghost number $\gh(\xi^a)\ge0$ and parity $\pa(\xi^a)$ equal to $0$
and $1$ for bosonic and fermionic fields respectively. Denote the
antifield corresponding to the field $\xi^a$ by $\xi^+_a$: it has
ghost number $\gh(\xi^+_a)=-1-\gh(\xi^a)<0$ and opposing parity
$\pa(\xi^+_a)=1-\pa(\xi^a)$ to $\xi^a$. Denote by $\CA$ the algebra of
functions in the variables $\xi^a$ and $\xi^+_a$ and their derivatives
\begin{equation*}
  \{\p^\ell\xi^a\}_{\ell\ge0}\cup \{\p^\ell\xi^+_a\}_{\ell\ge0}
\end{equation*}
with respect to the generators of negative ghost number. The bosonic
fields of ghost number $0$ play a special role in the theory: they are
coordinates on a manifold $M$, and we will view $\CA$ as a sheaf over
$M$.

The algebra $\CA$ is filtered by the ghost number of the
antifields. Let $F^k\CA$ be the ideal generated by monomials
\begin{equation*}
  \p^{\ell_1}\xi^+_{a_1} \dots \p^{\ell_n}\xi^+_{a_n}
\end{equation*}
such that $\gh(\xi^+_{a_1})+\dots+\gh(\xi^+_{a_n})+k\le0$. The
subspaces $F^k\CA$ define a decreasing filtration of $\CA$, with
$F^0\CA=\CA$ and $F^j\CA\cdot F^k\CA\subset F^{j+k}\CA$. Denote by
$\Ahat$ the completion of $\CA$ with respect to this filtration:
\begin{equation*}
  \Ahat = \varprojlim_k \CA/F^k\CA .
\end{equation*}

Introduce the partial derivatives
\begin{align*}
  \p_{k,a}
  &= \frac{\p~}{\p(\p^k\xi^a)} : \Ahat^j\to\Ahat^{j-\gh(\xi^a)} , &
  \p^a_k
  &= \frac{\p~}{\p(\p^k\xi^+_a)} : \Ahat^j\to\Ahat^{j-\gh(\xi^+_a)} .
\end{align*}
Let $\p$ be the total derivative with respect to $t$:
\begin{equation*}
  \p = \sum_{k=0}^\infty \Bigl( (\p^{k+1}\xi^a) \p_{k,a} +
  (\p^{k+1}\xi^+_a) \p^a_k \Bigr) .
\end{equation*}
An \textbf{evolutionary} vector field $X$ is a graded derivation of
$\CA$ that commutes with $\p$; such a vector field has the form
\begin{align*}
  X
  &= \pr\left( X^a \p_a + X_a \p^a \right) \\
  &= \sum_{k=0}^\infty \left( \p^kX^a \p_{k,a} + \p^kX_a \p_k^a \right) .
\end{align*}

The Soloviev antibracket on $\CA$ is defined by the formula
\begin{multline*}
  \[ f , g \] = \sum_a (-1)^{(\pa(f)+1)\pa(\xi^a)} \\
  \sum_{k,\ell=0}^\infty \bigl( \p^\ell ( \p_{a,k}f ) \, \p^k
  (\p^a_\ell g) + (-1)^{\pa(f)} \p^\ell ( \p^a_kf ) \, \p^k (
  \p_{a,\ell}g ) \bigr) .
\end{multline*}
This bracket, and its extension to $\Ahat$, satisfies the axioms for a
graded Lie superalgebra: it is graded antisymmetric
\begin{equation*}
  \[ y,x \] = - (-1)^{(\pa(x)+1)(\pa(y)+1)} \[ x,y \] ,
\end{equation*}
and satisfies the Jacobi relation
\begin{equation*}
  \[x,\[y,z\]\] = \[\[x,y\],z\] + (-1)^{(\pa(x)+1)(\pa(y)+1)} \[y,\[x,z\]\] .
\end{equation*}
Furthermore, it is linear over $\p$:
\begin{equation*}
  \[ \p f , g \] = \[ f , \p g \] = \p \[ f , g \] .
\end{equation*}
In this paper, all graded Lie superalgebras are $1$-shifted: the
antibracket has ghost number $1$.

The superspace $\CF=\Ahat/\p\Ahat$ of functionals is the graded
quotient of $\Ahat$ by the subspace $\p\Ahat$ of total
derivatives. Denote the image of $f\in\Ahat$ in $\CF$ by $\int f$. The
Soloviev antibracket $\[f,g\]$ descends to an antibracket
\begin{equation*}
  (\tint f , \tint g )
\end{equation*}
on $\CF$, called the Batalin--Vilkovisky antibracket. Thus, $\CF$ is a
sheaf of graded Lie superalgebras over $M$.

Introduce the variational derivatives
\begin{align*}
  \delta_a
  &= \sum_{k=0}^\infty (-\p)^k \p_{k,a} :
    \CF^j\to\Ahat^{j-\gh(\xi^a)} , &
  \delta^a
  &= \sum_{k=0}^\infty (-\p)^k \p^a_k :
    \CF^j\to\Ahat^{j-\gh(\xi^+_a)} .
\end{align*}
The Batalin--Vilkovisky antibracket is given by the formula
\begin{equation*}
  ( \tint f , \tint g ) = (-1)^{(\pa(f)+1)\pa(\xi^a)} \tint \bigl(
  (\delta_af ) \, (\delta^a g) + (-1)^{\pa(f)} ( \delta^af ) \, (
  \delta_ag ) \bigr) .
\end{equation*}

The classical Batalin--Vilkovisky master equation is the equation for
an element $S\in\CF$ of ghost number $0$ and even parity
\begin{equation}
  \label{MASTER}
  \half ( \tint S , \tint S ) = 0 .
\end{equation}
The Batalin--Vilkovisky differential is the Hamiltonian vector field
\begin{equation*}
  s = \sum_a (-1)^{\pa(\xi^a)} \pr \bigl( (\delta_aS)
  \p^a + (\delta^aS) \p_a \bigr) .
\end{equation*}
This is a graded derivation of ghost number $1$, and satisfies the
equation $s^2=0$ precisely when $S$ satisfies the classical master
equation \cite{covariant}*{Section~3}.

\section{The particle}

Before recalling the Batalin--Vilkovisky approach to the
superparticle, we review the simpler case of the particle. Consider
the $d$-dimensional Minkowski space $\V=\R^{d-1,1}$ with basis
$\{v_\mu\}_{0\le\mu<d}$ and inner product
\begin{equation*}
  \< v_\mu , v_\nu \> = \eta_{\mu\nu} .
\end{equation*}
The particle has physical fields $x^\mu$, and Lagrangian density
$S=\half \eta_{\mu\nu} \p x^\mu \p x^\nu$. For technical reasons, we
prefer to work in a first-order formulation of this theory, which has
additional physical fields $p_\mu$, and the action
\begin{equation*}
  S = p_\mu \p x^\mu - \half \eta^{\mu\nu} p_\mu p_\nu .
\end{equation*}

In order to have a theory with local reparametrization invariance, we
may couple the particle to ``gravity'' on the world-line, represented
by a nowhere-vanishing $1$-form field $e$, the graviton. Of course,
the gravitational field in dimension $1$ has no dynamical content. The
modified action for the particle is
\begin{equation*}
  S_{[0]} = p_\mu \p x^\mu - \half e \eta^{\mu\nu} p_\mu p_\nu .
\end{equation*}
The associated differential is
\begin{equation*}
  s_{[0]} = \pr \left( ( \p x^\mu - \eta^{\mu\nu} e p_\nu )
    \frac{\p~}{\p p^{+\mu}} - \p p_\mu \, \frac{\p~}{\p x^+_\mu} -
    \half \eta^{\mu\nu} p_\mu p_\nu \frac{\p~}{\p e^+} \right) .
\end{equation*}
The variation $\s_{[0]}e^+=-\half \eta^{\mu\nu} p_\mu p_\nu$ may be
recognized as the $d=1$ stress-energy tensor.

The local gauge symmetries of this model correspond to
cohomology classes of $s_{[0]}$ at ghost number $-1$:
\begin{equation*}
  s_{[0]} \bigl( \p e^+ - \eta^{\mu\nu} x^+_\mu p_\nu \bigr) = 0 .
\end{equation*}
This cohomology class is killed by the introduction of a ghost field
$c$, with ghost-number $1$, transforming as a scalar on the worldline,
and the addition to the action of the term
\begin{equation*}
  S_{[1]} = \bigl( \p e^+ - \eta^{\mu\nu} x^+_\mu p_\nu \bigr) c .
\end{equation*}
This adds the following terms to the differential:
\begin{equation*}
  s_{[1]} = \pr \left( \eta^{\mu\nu} c x^+_\nu \frac{\p~}{\p p^{+\mu}}
    + ( \p e^+ - \eta^{\mu\nu} x^+_\mu p_\nu ) \frac{\p~}{\p c^+}
    - \eta^{\mu\nu} c p_\nu \frac{\p~}{\p x^\mu} - \p c \frac{\p~}{\p
      e} \right) .
\end{equation*}

We see that the bosonic fields of ghost number $0$ of the theory are
the position $x^\mu$ and the momentum $p_\mu$, and the manifold
$\Phase$ is the cotangent bundle $T^\vee\V$ of $\V$. For definiteness,
we take the structure sheaf $\CO$ of $\Phase$ to be functions with
analytic dependence on $x^\mu$ and algebraic dependence on $p_\mu$,
but our results are actually insensitive to the regularity as
functions of $x^\mu$. The sheaf $\CA$ is the graded commutative
algebra generated over $\CO$ by the variables
\begin{equation*}
  \{\p^\ell x^\mu,\p^\ell p_\mu\}_{\ell>0} \cup \{\p^\ell
  e,e^{-1},\p^\ell c,\p^\ell x^+_\mu,\p^\ell p^{+\mu},\p^\ell
  e^+,\p^\ell c^+\}_{\ell\ge0} .
\end{equation*}
As in the last section, we denote its completion by $\Ahat$.

The sum $S=S_{[0]}+S_{[1]}$ satisfies the classical master equation,
and the cohomology of the differential $s=s_{[0]}+s_{[1]}$ on the
space of functionals $\CF$ vanishes below degree $-1$. 

In preparation for the proof, we recall a criterion of Boardman for
the convergence of a spectral sequence. Let $V$ be a complex, with
differential $d:V^i\to V^{i+1}$. A decreasing filtration on $V$ is a
sequence of subcomplexes
\begin{equation*}
  \dots \supset F^{-1}V \supset F^0V \supset F^1V \supset \cdots
\end{equation*}
The associated graded complex is
\begin{equation*}
  \gr^k_FV = F^kV/F^{k+1}V .
\end{equation*}
The filtration is \textbf{exhaustive} if for each $i\in\Z$,
\begin{equation*}
  \bigcup_k F^kV^i = V^i .
\end{equation*}
The filtration is \textbf{Hausdorff} if for each $i\in\Z$,
\begin{equation*}
  \bigcap_k F^kV^i = 0 .
\end{equation*}
The filtration is \textbf{complete} if
\begin{equation*}
  V = \varprojlim_k V/F^kV .
\end{equation*}
The filtration $F^\bull$ induces a filtration on the cohomology
$H^*(V)$, which we denote by the same letter. The spectral sequence
associated to the filtration converges if for all $(p,q)\in\Z^2$ the
induced morphism
\begin{equation*}
  \gr^p_F H^{p+q}(V) \longrightarrow E^{pq}_\infty
\end{equation*}
is an isomorphism, and the induced filtration on $H^*(V)$ is complete,
exhaustive and Hausdorff. The spectral sequence degenerates if
$E_\infty=E_r$ for $r\gg0$.
\begin{theorem}[Boardman \cite{Boardman}]
  If the spectral sequence associated to a complete, exhaustive
  Hausdorff filtration $(V,d,F^kV)$ degenerates, then it is
  convergent.
\end{theorem}
\begin{proof}
  Combine the following results from \cite{Boardman}: Theorems~8.2 and
  9.2, the remark after Theorem~7.1, and Lemma~8.1.
\end{proof}

A filtration on a differential graded algebra $A$ is a filtration on
the underlying complex such that $F^jA\cdot F^kA\subset F^{j+k}A$. In
this case, the pages $(E_r,d_r)$ of the spectral sequence are
themselves differential graded algebras, and the product on
$E_{r+1}\cong H^*(E_r,d_r)$ is induced by the product on $E_r$.

Introduce the light-cone
\begin{equation*}
  \Cone = \{ (x^\mu,p_\mu) \in \Phase \mid \eta^{\mu\nu}p_\mu p_\nu=0 \} .
\end{equation*}
\begin{theorem}
  \label{particle:vanish}
  The sheaf $H^i(\Ahat,s)$ vanishes for $i<0$, and is concentrated on
  the light-cone $\Cone$.

  Let $\tCA$ be the quotient of $\Ahat$ by constant multiples of the
  identity. The sheaf $H^i(\tCA,s)$ also vanishes for $i<0$, and is
  concentrated on the light-cone.
\end{theorem}
\begin{proof}
  Introduce an auxilliary grading on the sheaf of algebras $\CA$. The
  structure sheaf $\CO$ of the manifold $\Phase$ is placed in degree
  $0$, and the generators of $\CA$ over $\CO$ are assigned the degrees
  in the following table:
  \begin{equation*}
    \begin{tabu}{|c||c|c|c|c|c|c|c|c|c|} \hline
      \Phi & \p^\ell x^\mu & \p^\ell p_\mu & \p^\ell e & e^{-1} &
      \p^\ell c & \p^\ell x^+_\mu & \p^\ell p^{+\mu} & \p^\ell e^+ &
      \p^\ell c^+ \\ \hline
      \deg(\Phi) & 0 & 0 & 0 & 0 & 3 & 0 & 0 & -1 & -1 \\ \hline
    \end{tabu}
  \end{equation*}

  Write $\gh(f)=\gh_+(f)-\gh_-(f)$, where $\gh_+(f)$ and $\gh_-(f)$
  are the contributions of the fields, respectively antifields, to the
  ghost number. Rearranging, we see that
  \begin{equation*}
    \gh_-(f) = \gh_+(f) - \gh(f) .
  \end{equation*}
  Since
  \begin{equation*}
    \gh_+(f) + 2\gh(f) \le \deg(f) \le 3\gh_+(f) ,
  \end{equation*}
  we see that
  \begin{equation}
    \label{coerce}
     \tfrac{1}{3} \bigl( \deg(f) - 3 \gh(f) \bigr) \le \gh_-(f) \le
     \deg(f) - 3\gh(f) .
  \end{equation}
  From this grading, we construct an exhaustive and Hausdorff
  descending filtration on $\CA$: $G^k\CA^i$ is the span of elements
  $f\in\CA^i$ such that $\deg(f)\ge k$. By \eqref{coerce}, the
  completion of this filtration is isomorphic to $\Ahat$.

  We now consider the spectral sequence for the filtration induced by
  $G$ on $\Ahat$. We will show that $E^{pq}_\infty=E^{pq}_2$, that the
  sheaf $E^{pq}_2$ vanishes if $p+q<0$, and that it is supported on
  the light-cone $\Cone$. This establishes the theorem.

  The differential $s_0$ of the zeroth page $E^{pq}_0$ of the spectral
  sequence equals
  \begin{equation*}
    s_0 = \pr \left( \bigl( \p x^\mu - \eta^{\mu\nu} e p_\nu \bigr)
      \, \frac{\p~}{\p p^{+\mu}} - \p p_\mu \, \frac{\p~}{\p x^+_\mu}
      + \p e^+ \, \frac{\p~}{\p c^+} \right) .
  \end{equation*}
  This is a Koszul differential and its cohomology $E_1$ is the graded
  commutative algebra generated over $\CO$ by the variables
  \begin{equation*}
    \{\p^\ell e,e^{-1},\p^\ell c,e^+\}_{\ell\ge0} .
  \end{equation*}
  
  The differential $s_1$ of the first page $E_1$ of the spectral
  sequence equals
  \begin{equation*}
    s_1 = \pr \left( - \tfrac{1}{2} \eta^{\mu\nu} p_\mu p_\nu \,
      \frac{\p~}{\p e^+} \right) .
  \end{equation*}
  The element $\eta^{\mu\nu}p_\mu p_\nu\in E^{00}_1$ is not a zero
  divisor in $E_1$. We conclude that $E_2$ vanishes in negative
  degrees and is concentrated on the zero-locus of
  $\eta^{\mu\nu}p_\mu p_\nu$ in $\Phase$, namely the light-cone
  $\Cone$.

  We see that the second page $E_2$ of the spectral sequence is a
  graded commutative algebra, generated over $\CO_C$ by the variables
  \begin{align*}
    \{\p^\ell e,e^{-1}\}_{\ell\ge0} &\in E_2^{00} , &
    \{\p^\ell c\}_{\ell\ge0} &\in E_2^{3,-2} .
  \end{align*}
  We see that $s_2$ vanishes, so that $E_3^{pq}=E_2^{pq}$, and that
  the differential $s_3$ of the third page $E_3^{pq}$ equals
  \begin{equation*}
    s_3 = \pr \left( - \eta^{\mu\nu} c p_\mu \, \frac{\p~}{\p x^\nu} +
      \p c \, \frac{\p~}{\p e} \right) .
  \end{equation*}
  Using an auxilliary filtration, we see that the differentials $s_r$
  vanish for $r>3$, and $E_3^{pq}$ is quasi-isomorphic to the quotient
  complex obtained by taking the variables $\{\p^\ell e,\p^\ell
  c\}_{\ell>0}$ to $0$ and the variables $e$ and $e^{-1}$ to $1$:
  \begin{equation*}
    \begin{tikzcd}[column sep=2.7em]
      0 \ar{r} & \CO_C \ar{rrr}{\textstyle-\eta^{\mu\nu}cp_\mu\p/\p
        x^\nu} & & &
      c\,\CO_C \ar{r} & 0 .
    \end{tikzcd}
  \end{equation*}

  Turning to the case of the sheaf $\tCA$, we have a long exact
  sequence for cohomology sheaves
  \begin{equation*}
    0 \longrightarrow H^{-1}(\tCA,s) \longrightarrow \R
    \longrightarrow H^0(\Ahat,s) \longrightarrow \cdots
  \end{equation*}
  But the above proof shows that the morphism $\R\to E^{00}_\infty$ is
  an injection, and hence that $H^{-1}(\tCA,s)=0$.
\end{proof}

\begin{corollary}
  \label{particle:main}
  Let $\CF=\Ahat/\p\Ahat$. The cohomology sheaf $H^i(\CF,\s)$ vanishes
  for $i<-1$, and is concentrated on the light-cone $\Cone$.
\end{corollary}
\begin{proof}
  The sheaf $\CF$ has a resolution
  \begin{equation*}
    0 \longrightarrow \tCA \overset{\p}{\longrightarrow}
    \Ahat \longrightarrow \CF \longrightarrow 0 .
  \end{equation*}
  The associated long exact sequence implies that $H^i(M,\CF) = 0$ for
  $i<-1$.
\end{proof}

In \cite{covariant}, we show that the covariance of a field theory
with respect to reparametrization of the worldline may be expressed by
introducing the graded Lie superalgebra $\CF[[u]]$ of power series in
a formal variable $u$ of ghost number $2$ and even parity. Consider
the global section of $\CA$
\begin{equation*}
  D = x^+_\mu \p x^\mu + p^{+\mu} \p p_\mu - e \p e^+ + c^+ \p c ,
\end{equation*}
of ghost number $-1$ and odd parity. Its associated Hamiltonian vector
field is $\p$, and its image $\tint D$ in $\CF$ is central. Let
$G=x^+_\mu p^{+\mu} + e c^+$. Covariance of the theory is expressed by
the curved Maurer--Cartan equation
\begin{equation*}
  \half ( \tint S_u , \tint S_u ) = - u \tint D ,
\end{equation*}
where
\begin{equation*}
  S_u = S + u G .
\end{equation*}
If $\int f\in\CF^k$ is a cocycle, then
\begin{equation*}
  (\tint G,\tint f)\in\CF^{k-1}
\end{equation*}
is again a cocycle, called the \textbf{transgression} of $f$. In
particular, the long exact sequence
\begin{equation*}
  \begin{tikzcd}
    \cdots \arrow{r}{\p} & H^{-1}(\Ahat,s) \arrow{r} & H^{-1}(\CF,s)
    \arrow{r} & H^0(\tCA,s)
    \arrow[out=-5,in=170,overlay]{lld}[']{\p} & \\
    & H^0(\Ahat,s) \arrow{r} & H^0(\CF,s) \arrow{r} & H^1(\tCA,s)
    \arrow{r}{\p} & \cdots
  \end{tikzcd}
%  \begin{tikzcd}
%    \cdots \arrow{r} & H^{-1}(\tCA,s) \arrow{r}{\p} & H^{-1}(\CA,s)
%    \arrow{r} & H^{-1}(\CF,s) \arrow[out=-5,in=170,overlay]{lld} & \\
%    & H^0(\tCA,s) \arrow{r}{\p} & H^0(\CA,s) \arrow{r} & H^0(\CF,s)
%    \arrow[out=-5,in=170,overlay]{lld} & \\
%    & H^1(\tCA,s) \arrow{r}{\p} & H^1(\CA,s) \arrow{r} & H^1(\CF,s)
%    \arrow{r} & \cdots
%  \end{tikzcd}
\end{equation*}
splits, in the sense that the morphisms $\p$ vanish.

The particle is actually an AKSZ model (Alexandrov et al.\
\cite{AKSZ}), associated to the symplectic supermanifold
$T^\vee(\V\times\R[1])$, and $G$ may be interpreted as the Poisson
tensor on this supermanifold. The main result of this paper is to find
the analogue of $S_u$ for the superparticle, establishing that the
superparticle is a ``generalized AKSZ model.''

\section{The superparticle}

The Batalin--Vilkovisky action $\S$ for the superparticle is based on
the above action for the particle, in the special case where
$\V=\R^{9,1}$. Recall some properties of Majorana--Weyl spinors in
this signature of space-time: for further details, see the
Appendix. The spin group $\Spin(9,1)$ is the universal cover of the
identity component of $\SO(9,1)$. It has two real irreducible
sixteen-dimensional representations, the left and right-handed
Majorana--Weyl spinors $\Ss _+$ and $\Ss_-$. The $\gamma$-matrices
$\gamma^\mu:\Ss_\pm\to\Ss_\mp$ satisfy the relations
\begin{equation*}
  \gamma^\mu\gamma^\nu + \gamma^\nu\gamma^\mu = 2\eta^{\mu\nu} .
\end{equation*}
The Lie algebra of the group $\Spin(9,1)$ is spanned by the quadratic
expressions
\begin{equation*}
  \gamma^{\mu\nu} = \half \bigl( \gamma^\mu\gamma^\nu -
  \gamma^\nu\gamma^\mu \bigr) .
\end{equation*}
There is a non-degenerate symmetric bilinear form $\T(\alpha,\beta)$
on $\Ss$, which vanishes on $\Ss_\pm\otimes\Ss_\pm$ and places $\Ss_\pm$ in
duality with $\Ss_\mp$. We have
\begin{equation*}
  \T^\mu(\alpha,\beta) = \T( \gamma^\mu\alpha , \beta ) = \T( \alpha ,
  \gamma^\mu \beta ) .
\end{equation*}
In particular, we see that
\begin{equation*}
  \T( \gamma^{\mu\nu}\alpha , \beta ) = - \T( \alpha ,
  \gamma^{\mu\nu} \beta ) .
\end{equation*}
Hence, the pairing $\T(\alpha,\beta)$ is invariant under the action of
the Lie group $\Spin(9,1)$, and, in particular,
$\Ss_-\cong(\Ss_+)^\vee$ as a representation of $\Spin(9,1)$.

To define the superparticle, we adjoin to the particle a series of
fields $\theta_n$, $n\ge0$, of ghost number $n$, which are left-handed
Majorana--Weyl spinors if $n$ is even, and right-handed Majorana--Weyl
spinors if $n$ is odd: the parity of $\theta_n$ is the opposite of the
parity of $n$. As functions on the worldline, these fields all
transform as scalars.

For the correct definition of the superparticle, it is necessary to
exclude the states of vanishing momentum. To this end, we let
$\Phase_0$ be the complement in $\Phase=T^\vee\V$ of the
zero-section. Denote by $j: M_0 \to M$ the open embedding, and by
$\CO_0=j^*\CO$ the structure sheaf of $\Phase_0$. Denote by $\AA$ the
algebra generated over $\CA_0=j^*\CA$ by the variables
\begin{equation*}
  \{\p^\ell\theta_n,\p^\ell\theta^+_n \mid n\ge0 \}_{\ell\ge0} .
\end{equation*}
We denote the completion of $\AA$ with respect to antifields by
$\AAhat$.

We see that $\AA$ is a sheaf of graded commutative algebras over
$\Phase_0$. Let
\begin{equation*}
  \Cone_0 = \Cone \cap \Phase_0
\end{equation*}
be the intersection of the light-cone with open submanifold $\Phase_0$
of the cotangent bundle.

Consider the composite fields
\begin{equation}
  \label{Psi}
  \Psi_n =
  \begin{cases}
    (-1)^{\binom{n+1}{2}} \, \theta^+_{-n-1} , & n<-1 , \\[3pt]
    \theta^+_0 + \half x^+_\mu \gamma^\mu \theta_0 + 2c^+\theta_1 , &
    n=-1 , \\[3pt]
    \p \theta_n + (-1)^{n+1} x^+_\mu \gamma^\mu \theta_{n+1} +
    2c^+\theta_{n+2} , & n\ge0 .
  \end{cases}
\end{equation}

Denote by $\S$ the full Batalin--Vilkovisky action of the
superparticle, and by $\s$ the associated Batalin--Vilkovisky
differential. The formula for $\S$ may be found in Lindstr\"om et al.\
\cite{LRSVV}: we content ourselves here with the following
characterization, in terms of the differential $\s$.

\begin{proposition}
  The Batalin--Vilkovisky action $\S$ of the superparticle is
  characterized by the following conditions:
  \begin{enumerate}[i)]
  \item $\S$ satisfies the classical master equation;
  \item $\S=S+S'$, where $S$ is the Batalin--Vilkovisky action of the
    particle
    \begin{equation*}
      S = p_\mu \p x^\mu - \half e \eta^{\mu\nu} p_\mu p_\nu + \bigl(
      \p e^+ - \eta^{\mu\nu} x^+_\mu p_\nu \bigr) c ,
    \end{equation*}
    and $S'$ depends only on the fields and antifields
    $\{p_\mu,\theta_n,x^+_\mu,e^+,c^+,\theta^+_n\}$ and their
    derivatives;
  \item for all $n\in\Z$, the differential $\s$ acts on the composite
    fields $\Psi_n$ as follows:
    \begin{equation*}
      \s\Psi_n = (-1)^{n+1} p_\mu \gamma^\mu \Psi_{n+1} - 2e^+
      \Psi_{n+2} .
    \end{equation*}
  \end{enumerate}
\end{proposition}

Using the vanishing result Theorem \ref{particle:vanish}, we now
establish the analogous result for the superparticle.
\begin{theorem}
  \label{thm:vanish}
  The sheaf $H^i(\AAhat, \s)$ vanishes for $i<0$, and is concentrated
  on the light-cone $\Cone_0$.

  Let $\tAA$ be the quotient of $\AAhat$ by constant multiples of the
  identity. The sheaf $H^i(\tAA,\s)$ also vanishes for $i<0$, and is
  concentrated on the light-cone $\Cone_0$
\end{theorem}

In the proof of Theorem \ref{thm:vanish}, we need the formula for the
differential $\s$ on fields and antifields of the theory. We see that
$\s$ equals $s$ on the fields and antifields
$\{p_\mu,x^+_\mu,e^+,c^+\}$, and
\begin{multline*}
  \s \theta_n = (-1)^{n+1} p_\mu \gamma^\mu \theta_{n+1} - 2e^+
  \theta_{n+2} \\
  \shoveleft{\s x^\mu = - \eta^{\mu\nu} cp_\nu + \half p_\nu
    \T(\gamma^\nu\gamma^\mu\theta_0,\theta_1) + e^+
    \T^\mu(\theta_1,\theta_1) - e^+ \T^\mu(\theta_0,\theta_2)} \\
  \shoveleft{\s c = - p_\mu \T^\mu(\theta_1,\theta_1) - 4e^+
    \T(\theta_1,\theta_2)} \\
  \shoveleft{\s e = -\p c + x^+_\mu \T^\mu(\theta_1,\theta_1) - 4 c^+
    \T(\theta_1,\theta_2) + 2 \sum_{n=0}^\infty (-1)^{\binom{n}{2}}
    \T(\Psi_{-n},\theta_{n+1})} \\
  \shoveleft{\s p^{+\mu} = \p x^\mu - \eta^{\mu\nu} ep_\nu + \half
    x^+_\nu \T(\gamma^\nu\gamma^\mu\theta_0,\theta_1)
    - c^+ \T^\mu(\theta_1,\theta_1) + c^+ \T^\mu(\theta_0,\theta_2) } \\
  + \half \T^\mu(\Psi_0,\theta_0) + \sum_{n=1}^\infty
  (-1)^{\binom{n}{2}} \T^\mu(\Psi_{-n},\theta_n) .
\end{multline*}
The infinite sums in the formulas for $\s e$ and $\s p^{+\mu}$ make
sense by the completeness property of $\AAhat$.

\begin{proof}[Proof of Theorem \ref{thm:vanish}]
  We define an auxilliary grading on $\AA$ extending the grading on
  $\CA_0$ used in the proof of Theorem~\ref{particle:vanish}: the
  generators of $\AA$ over $\CA_0$ are assigned the degrees in the
  following table:
  \begin{equation*}
    \begin{tabu}{|c||c|c|c|} \hline
      \Phi & \theta_n & \p^\ell\Psi_n & \p^\ell\Psi_{-n} \\ \hline
      \deg(\Phi) & 3n+1 & 3n & -2n \\ \hline
    \end{tabu}
  \end{equation*}
  Observe that
  \begin{equation*}
    \deg(f) \le 4\gh_+(f) + 16 ;
  \end{equation*}
  the factor $4$ accounts for the field $\theta_1$, which has ghost
  number $1$ and degree $4$, while the constant $16$ accounts for the
  $16$ modes of the fermionic field $\theta_0$, which have ghost
  number $0$ and degree $1$. In the other direction, we have
  \begin{equation*}
    \gh_+(f) + 2 \gh(f) \le \deg(f) .
  \end{equation*}
  Combining these two inequalities, we see that
  \begin{equation}
    \label{coerce:superparticle}
    \tfrac{1}{4} \deg(f) - \gh(f) - 4 \le \gh_-(f) \le \deg(f) -
    3\gh(f) .
  \end{equation}

  From this grading, we construct an exhaustive and Hausdorff
  descending filtration on $\AA$: $G^k\AA^i$ is the span of elements
  $f\in\AA^i$ such that $\deg(f)\ge k$. By
  \eqref{coerce:superparticle}, the completion of this filtration is
  isomorphic to $\AAhat$.

  The differential $\s_0$ on the zeroth page of the spectral sequence
  $E^{pq}_0$ equals
  \begin{equation*}
    \s_0 = \pr \left( \p x^\mu \, \frac{\p~}{\p p^{+\mu}}
      - \p p_\mu \, \frac{\p~}{\p x^+_\mu}
      - \p e^+ \, \frac{\p~}{\p c^+} \right) .
  \end{equation*}
  This is a Koszul differential and its cohomology $E_1$ is the graded
  commutative algebra generated over $\CO_0$ by the variables
  \begin{equation*}
    \{ \p^\ell e, e^{-1}, \p^\ell c , e^+ \}_{\ell\ge0} \cup \{
    \theta_n \mid n\ge0 \} \cup \{ \p^\ell\Psi_n \mid n\in\Z
    \}_{\ell\ge0} .
  \end{equation*}
  
  The differential $\s_1$ on the first page $E_1$ of the spectral
  sequence is given by the formula
  \begin{equation*}
    \s_1 = \pr \left( - \tfrac{1}{2} \eta^{\mu\nu} p_\mu p_\nu \,
      \frac{\p~}{\p e^+} \right) .
  \end{equation*}
  The element $\eta^{\mu\nu}p_\mu p_\nu$ is not a zero divisor in
  $E_1$: its zero-locus in $\Phase_0$ is the light-cone $\Cone_0$,
  with structure sheaf $\CO_{\Cone_0}$.

  We see that the second page $E_2$ of the spectral sequence is a
  graded commutative algebra generated over $\CO_{\Cone_0}$ by the
  variables
  \begin{equation*}
    \{ \p^\ell e , e^{-1} , \p^\ell c \}_{\ell\ge0} \cup \{
      \theta_n \mid n\ge0 \} \cup \{ \p^\ell\Psi_n \mid n\in\Z
      \}_{\ell\ge0} .
  \end{equation*}
  The differential $\s_2$ on the second page $E_2$ of the spectral
  sequence is given by the formula
  \begin{equation*}
    \s_2 = \sum_{n=1}^\infty (-1)^{n+1} \pr \left( p_\mu \T^\mu\left(
        \Psi_{1-n} , \frac{\p~}{\p\Psi_{-n}} \right) \right) .
  \end{equation*}
  On the light-cone $\Cone_0$, the operator
  \begin{equation*}
    p_\mu \gamma^\mu : \Ss_\pm \to \Ss_\mp
  \end{equation*}
  has square zero, since
  $(p_\mu\gamma^\mu)^2=\eta^{\mu\nu}p_\mu p_\nu=0$. The cohomology of
  this operator vanishes, in the sense that
  \begin{equation*}
    \ker(p_\mu \gamma^\mu) = \im(p_\mu\gamma^\mu) .
  \end{equation*}
  To see this, choose a vector $q_\mu$ such that
  $\eta^{\mu\nu}p_\mu q_\nu>0$: then $q_\mu\gamma^\mu$ yields a
  contracting homotopy for the differential $p_\mu\gamma^\mu$. (This
  is where in the proof we need to have localized away from the zero
  section of $\Phase$.)

  The third page $E^{pq}_3$ is generated over $\CO_{C_0}$ by the
  variables
  \begin{align*}
    \{\p^\ell e,e^{-1}\}_{\ell\ge0} & \in E^{00}_3 , &
    \{\theta_n \mid n\ge0\} & \in E^{3n+1,-2n-1}_3, \\
    \{\p^\ell c\}_{\ell\ge0} &\in E^{3,-2}_3 , &
    \{\p^\ell\Psi_n\mid n\ge0\}_{\ell\ge0} &\in E^{3n,-2n}_3 ,
  \end{align*}
  modulo relations
  \begin{equation*}
    \{ \p^\ell ( p_\mu\gamma^\mu\Psi_0 ) \}_{\ell\ge0} \in
    E^{00}_3 .
  \end{equation*}
  Thus, $E^{pq}_r$ vanishes unless $p\ge0$, $p+q\ge0$, and
  $3p+4q\ge-16$; this last inequality is saturated by the product of
  the $16$ components of the field $\theta_0$, located in
  $E^{16,-16}_0$. The differential $\s_r$ of the $r$th page of the
  spectral sequence vanishes for $r>20$, and hence the spectral
  sequence degenerates, proving the first part of the theorem.

  The proof of the vanishing of the cohomology sheaves $H^i(\tAA,\s)$
  follows the same lines as the proof of the analogous result for the
  particle.
\end{proof}

\begin{corollary}
  \label{superparticle:main}
  Let $\FF=\AAhat/\p\AAhat$. The sheaf $H^i(\FF,\s)$ vanishes for
  $i<-1$, and is concentrated on the light-cone $\Cone_0$.
\end{corollary}

\section{The Thom--Whitney normalization}

Let $X$ be a manifold with cover
\begin{equation*}
  \CU =  \{ U_\alpha \}_{\alpha\in I} .
\end{equation*}
The nerve $N_k\CU$ of the cover is the sequence of manifolds indexed
by $k\ge0$
\begin{equation*}
  N_k\CU = \bigsqcup_{\alpha_0\dots\alpha_k\in I^{k+1}}
  U_{\alpha_0\dots\alpha_k} ,
\end{equation*}
where
\begin{equation*}
  U_{\alpha_0\dots\alpha_k} = U_{\alpha_0} \cap \dots \cap U_{\alpha_k} .
\end{equation*}
Denote by $\epsilon:N_0\CU\to X$ the map which on each summand
$U_\alpha$ equals the inclusion $U\hookrightarrow X$.

In order to globalize the classical master equation, we have to
replace the manifold $X$ by a sequence of manifolds of the form
$\{N_k\CU\}$. To do this, we will use the formalism of simplicial and
cosimplicial objects, and we now review their definition.

Let $\Delta$ be the category whose objects are the totally ordered
sets
\begin{equation*}
  [k] = (0<\dots<k) , \quad k\in\N ,
\end{equation*}
and whose morphisms are the order-preserving functions. A simplicial
manifold $M_\bull$ is a contravariant functor from $\Delta$ to the
category of manifolds. (We leave open here whether we are working in
the smooth, analytic or algebraic setting.)  Here, $M_k$ is the value
of $M_\bull$ at the object $[k]$, and $f^*:M_\ell\to M_k$ is the
action of the arrow $f:[k]\to[\ell]$ of $\Delta$. The arrow
$d_i:[k]\to[k+1]$ which takes $j<i$ to $j$ and $j\ge i$ to $j+1$ is
known as a face map, while the arrow $s_i:[k]\to[k-1]$ which takes
$j\le i$ to $j$ and $j>i$ to $j-1$ is known as a degeneracy map.

The simplicial manifolds used in this paper are the \Cech\ nerves
$N_\bull\CU$ of covers $\CU=\{U_\alpha\}_{\alpha\in I}$. The face map
$\delta_i=d_i^*:N_{k+1}\CU\to N_k\CU$ corresponds to the inclusion of
the open subspace
\begin{equation*}
  U_{\alpha_0\dots\alpha_{k+1}} \subset N_{k+1}\CU
\end{equation*}
into the open subspace
\begin{equation*}
  U_{\alpha_0\dots\widehat{\alpha}_i\dots\alpha_{k+1}} \subset N_k\CU ,
\end{equation*}
and the degeneracy map $\sigma_i=s_i^*:N_{k-1}U\to N_kU$ corresponds
to the identification of the open subspace
\begin{equation*}
  U_{\alpha_0\dots\alpha_k} \subset N_k\CU
\end{equation*}
with the open subspace
\begin{equation*}
  U_{\alpha_0\dots\alpha_i\alpha_i\dots\alpha_{k+1}} \subset N_{k+1}\CU .
\end{equation*}
Any simplicial map $f^*:M_\ell\to M_k$ is the composition of a
sequence of face maps followed by a sequence of degeneracy maps. In
particular, we see that in the case $M_\bull=N_\bull\CU$ of the nerve
of a cover, all of these maps are local embeddings.

A covariant functor $X^\bull$ from $\Delta$ to a category $\CC$ is
called a cosimplicial object of $\CC$. These arise as the result of
applying a contravariant functor to a simplicial space: for example,
applying the functor $\CF(-)$ to the simplicial graded supermanifold
$N_\bull\CU$, we obtain the cosimplicial graded Lie superalgebra
\begin{equation*}
  \CF(N_\bull\CU)
\end{equation*}
with the Batalin--Vilkovisky antibracket.

We now generalize the classical master equation of Batalin--Vilkovisky
theory to a Maurer--Cartan equation for the cosimplicial graded Lie
superalgebra $\CF(N_\bull\CU)$. We use a construction introduced in
rational homotopy theory by Sullivan \cite{Sullivan} (see also
Bousfield and Guggenheim \cite{BG}), the Thom--Whitney normalization.

Let $\Om_k$ be the free graded commutative algebra with generators
$\{t_i\}_{i=0}^k$ of degree $0$ and $\{dt_i\}_{i=0}^k$ of degree $1$,
and relations
\begin{equation*}
  t_0 + \dotsb + t_k = 1
\end{equation*}
and $dt_0+\dotsb+dt_k=0$. There is a unique differential $\delta$ on
$\Om_k$ such that $\delta(t_i)=dt_i$, and $\delta(dt_i)=0$.

The differential graded commutative algebras $\Om_k$ are the
components of a simplicial differential graded commutative algebra
$\Om_\bull$ (that is, contravariant functor from $\Delta$ to the
category of differential graded commutative algebras): the arrow
$f:[k]\to[\ell]$ in $\Delta$ acts by the formula
\begin{equation*}
  f^*t_i = \sum_{f(j)=i} t_j , \quad 0\le i\le n .
\end{equation*}

The Thom--Whitney normalization $\Om_\bull \otimes_\Delta V^\bull$ of
a cosimplicial superspace is the equalizer of the maps
\begin{equation*}
  \begin{tikzcd}
    \displaystyle
    \prod_{k=0}^\infty \Om_k \otimes V^k
    \arrow[shift right]{r}[']{1\otimes f_*}
    \arrow[shift left]{r}{f^*\otimes1} &
    \displaystyle
      \prod_{k,\ell=0}^\infty \prod_{f:[k]\to[\ell]} \Om_k \otimes V^\ell
  \end{tikzcd}
\end{equation*}
If the superspaces $V^k$ making up the cosimplicial superspace are
themselves graded $V^{k\ast}$, the Thom--Whitney totalization of
$V^{\bull\ast}$ is the product superspace
\begin{equation*}
  \|V\|^n = \prod_{k=0}^\infty \Om^k_\bull \otimes_\Delta V^{\bull,n-k} .
\end{equation*}

The Thom--Whitney normalization takes cosimplicial $1$-shifted graded
Lie superalgebras to $1$-shifted graded Lie superalgebras. The reason
is simple: if $L^k$ is a $1$-shifted graded Lie superalgebra, then so
is $\Om_k \otimes L^k$, with differential $\delta$ and antibracket
\begin{equation*}
  [\alpha_1\otimes x_1,\alpha_2\otimes x_2] = (-1)^{j_2\pa(x_1)+1} \,
  \alpha_1\alpha_2 \, [x_1,x_2] ,
\end{equation*}
where $\alpha_\ell\in\Om^{i_\ell}_k$ and $x_\ell\in L^{k,j_\ell}$. The
Thom--Whitney totalization $\|L\|$ is a subspace of the product of
$1$-shifted graded superalgebras $\Om_k\otimes L^k$, and this subspace is
preserved by the differential and by the antibracket.

The construction of $\|\CF(N_\bull\CU)\|$ behaves well under refinement
of covers. A refinement $\CV=\{V_\beta\}_{\beta\in J}$ of a cover
$\CU=\{U_\alpha\}_{\alpha\in I}$ is determined by a function of
indexing sets $\phi:J\to I$, such that for all $\beta\in J$, $V_\beta$
is a subset of $U_{\phi(\beta)}$. There is a morphism of cosimplicial
$1$-shifted graded Lie superalgebras
$\Phi^*:\CF(N_\bull\CU)\to\CF(N_\bull\CV)$, obtained by sections of
$\CF$ on $U_{\phi(\alpha_0)\dots\phi(\alpha_k)}$ to sections on
$V_{\alpha_0\dots\alpha_k}$. Applying the totalization functor, we
obtain a morphism of complexes
\begin{equation*}
  \Phi^* : \|\CF(N_\bull\CU)\| \to \|\CF(N_\bull\CV)\| .
\end{equation*}
If we have a further refinement $\CW=\{W_\gamma\}_{\gamma\in K}$ of
$\CV=\{V_\beta\}_{\beta\in J}$ with $\psi:K\to J$, we may define a
composition of these refinements $\phi\psi:K\to I$, and we obtain a
commuting triangle of morphisms of complexes
\begin{equation*}
  \begin{tikzcd}
    & \|\CF(N_\bull\CU)\| \arrow{dl}[']{\Phi^*}
    \arrow{dr}{\Psi^*\Phi^*} \\
    \|\CF(N_\bull\CV)\| \arrow{rr}[']{\Psi^*} & & \|\CF(N_\bull\CW)\|
  \end{tikzcd}
\end{equation*}
The arrows in this diagram are morphisms of differential graded
$1$-shifted Lie superalgebras.

The analogue of the classical master equation \eqref{MASTER} in the
global setting is the Maurer--Cartan equation for the differential
graded $1$-shifted Lie superalgebra $\|\CF(N_\bull\CU)\|$:
\begin{equation*}
  \delta \tint \S + \half ( \tint \S , \tint \S ) = 0 .
\end{equation*}
Here, $\S$ is a consistent collection of elements
$\S^j_{\alpha_0\dots\alpha_k} \in \Om^j_k \otimes
\CF^{-j}(U_{\alpha_0\dots\alpha_k})$ of total degree $0$ which
satisfies the sequence of Maurer--Cartan equations
\begin{equation*}
  \delta\tint\S^{j-1} + \frac12 \sum_{i=0}^j ( \tint \S^i , \tint
  \S^{j-i} ) = 0 .
\end{equation*}

\section{The superparticle as a covariant field theory}

We now return to the superparticle. In this section, using the
Thom--Whitney formalism of the previous section, we will show that the
superparticle is a global covariant field theory, in the terminology
of \cite{covariant}.

Let $\D\in\Gamma(\Phase_0,\AAhat^{-1})$ be the element
\begin{equation*}
  \D = x^+_\mu \p x^\mu + p^{+\mu} \p p_\mu - e \p e^+ + c^+ \p c
  + \sum_{n=0}^\infty \T( \theta^+_n , \p\theta_n) .
\end{equation*}
\begin{definition}
  A global covariant field theory is a solution of the curved
  Maurer--Cartan equation in $\|\FF(N_\bull\CU)[u]\|$, where $\CU$ is
  a cover of $\Phase_0$:
  \begin{equation*}
    \delta \tint \S_u + \half ( \tint \S_u , \tint \S_u ) = - u \tint
    \D .
  \end{equation*}
\end{definition}

If $\S_u$ is a covariant field theory with respect to a cover $\CU$ of
$\Phase_0$ and $(\CV,\phi)$ is a refinement of $\CU$, $\Phi^*\S_u$ is
again a global covariant field theory with respect to the refined
cover.

\begin{theorem}
  \label{main}
  There is a global covariant field theory
  \begin{equation*}
    \S_u = \S + \sum_{n=0}^\infty u^{n+1} \G_n
  \end{equation*}
  such that $\S$ is the solution of the classical master equation for
  the superparticle.
\end{theorem}
\begin{proof}
  Consider the open affine cover $\CU = \{U_\mu\}_{0\le\mu\le9}$ of
  $\Phase_0$, where
  \begin{equation*}
    U_\mu = \{p_\mu \ne 0 \} .
  \end{equation*}
  We must construct a series of cochains
  \begin{equation*}
    \tint \G_n \in \| \FF(N_\bull\CU) \|^{-2n-2} ,
  \end{equation*}
  in the Thom--Whitney totalization $\|\FF(N_\bull\CU)\|$ of the
  cosimplicial graded Lie superalgebra $\FF(N_\bull\CU)$, satisfying
  the curved Maurer--Cartan equation
  \begin{equation*}
    \delta \tint \S_u + \half ( \tint \S_u , \tint \S_u ) = - u \tint \D .
  \end{equation*}
  Equivalently, we must find a solution $\G_0$ of the equation
  \begin{equation}
    \label{G0}
    (\delta+\s)\tint\G_0 = - \tint\D ,
  \end{equation}
  and for $n>0$, solutions of the equations
  \begin{equation}
    \label{Gn}
    (\delta+\s) \tint \G_n =
    - \frac{1}{2} \sum_{j+k=n-1} ( \tint \G_j , \tint \G_k ) .
  \end{equation}
  Assuming that we have solved these equations for
  $(\G_0,\dots,\G_{n-1})$, we see that
  \begin{multline*}
    \frac{1}{2} \sum_{j+k=n-1} (\delta+\s) ( \tint \G_j , \tint \G_k ) \\
    = - ( \tint \D , \tint \G_{n-1} ) - \sum_{i+j+k=n-2} \bigl( ( \tint
    \G_i , \tint \G_j ) , \tint \G_k \bigr) .
  \end{multline*}
  The first term vanishes since $\tint\D$ lies in the centre of $\FF$,
  while the second term vanishes by the Jacobi relation for graded Lie
  superalgebras. Thus, the right-hand side of \eqref{Gn} is a
  cocycle. Since the cohomology of the complex $\| \FF(N_\bull\CU) \|$
  vanishes below degree $-1$ by Theorem~\ref{thm:vanish}, we may solve
  the equation for $\G_n$.

  Rewrite the formula for $\D$, using the definition \eqref{Psi} of
  $\Psi_n$ and the formula for the action of $\s$:
  \begin{equation}
    \label{DD}
    \D = - \s ( x^+_\mu p^{+\mu} + e c^+ ) + \frac{1}{2}
    \sum_{n=-\infty}^\infty (-1)^{\binom{n}{2}} \,
    \T(\Psi_{-n},\Psi_{n-1}) .
  \end{equation}

  Introduce the vector
  \begin{equation*}
    q_\mu = \frac{t_\mu}{2\eta^{\mu\nu}p_\nu} ,
  \end{equation*}
  and its de~Rham differential $\delta q_\mu$. We will show that the
  expression
  \begin{multline}
    \label{G0:explicit}
    \G_0 = x^+_\mu p^{+\mu} + e c^+ \\
    + \frac{1}{2} \sum_{k\ge0} \sum_{\nu_0\dots\nu_k} (-1)^k \,
    q_{\nu_0} \delta q_{\nu_1} \dots \delta q_{\nu_k}
    \sum_{n=-\infty}^\infty (-1)^{\binom{n}{2}} \,
    \T^{\nu_0\dots\nu_k}( \Psi_{-n}, \Psi_{n-k-2} )
  \end{multline}
  in $\|\AA(N_\bull\CU)\|^{-2}$ gives a solution of the equation
  \begin{equation}
    \label{GG0}
    (\delta+\s)\G_0 = - \D ,
  \end{equation}
  yielding \eqref{G0}. By \eqref{DD}, it suffices to show that
  \begin{multline}
    \label{GGG0}
    \s \sum_{\nu_0\dots\nu_k}q_{\nu_0} \delta q_{\nu_1} \dots \delta
    q_{\nu_k} \sum_{n=-\infty}^\infty (-1)^{\binom{n}{2}}
    \, \T^{\nu_0\dots\nu_k}( \Psi_{-n},\Psi_{n-k-2} ) \\
    =
    \begin{cases}
      \displaystyle
      \sum_{\nu_0\dots\nu_{k-1}} \delta q_{\nu_0} \dots \delta
      q_{\nu_{k-1}} \sum_{n=-\infty}^\infty (-1)^{\binom{n}{2}}
      \, \T^{\nu_0\dots\nu_{k-1}}( \Psi_{-n},\Psi_{n-k-1} ) , & k>0 ,
      \\[10pt]
      \displaystyle
      - \sum_{n=-\infty}^\infty (-1)^{\binom{n}{2}} \,
      \T(\Psi_{-n},\Psi_{n-1}) , & k=0 .
    \end{cases}
  \end{multline}
  We have
  \begin{multline*}
    \s \sum_{n=-\infty}^\infty (-1)^{\binom{n}{2}} \,
    \T^{\nu_0\dots\nu_k}( \Psi_{-n},\Psi_{n-k-2} ) \\
    \begin{aligned}
      &= \sum_{n=-\infty}^\infty (-1)^{\binom{n}{2}} \,
      \T^{\nu_0\dots\nu_k}( (-1)^{n+1} p_\mu \gamma^\mu \Psi_{-n+1} -
      2 e^+ \Psi_{-n+2} , \Psi_{n-k-2} ) \\
      &+ \sum_{n=-\infty}^\infty (-1)^{\binom{n}{2}+n+1} \,
      \T^{\nu_0\dots\nu_k}( \Psi_{-n}, (-1)^{n+k+1} p_\mu \gamma^\mu
      \Psi_{n-k-1} - 2e^+ \Psi_{n-k} ) \\
      &= p_\mu \sum_{n=-\infty}^\infty (-1)^{\binom{n}{2}} \Bigl(
      \T^{\nu_0\dots\nu_k}( \gamma^\mu \Psi_{-n}, \Psi_{n-k-1} ) \\
      & \qquad\qquad\qquad + (-1)^{k} \, \T^{\nu_0\dots\nu_k}(
      \Psi_{-n}, \gamma^\mu \Psi_{n-k-1} ) \Bigr) \\
      &- 2 e^+ \sum_{n=-\infty}^\infty \Bigl( (-1)^{\binom{n+2}{2}}
      + (-1)^{\binom{n}{2}} \Bigr) \T^{\nu_0\dots\nu_k}(
      \Psi_{-n}, \Psi_{n-k} ) .
    \end{aligned}
  \end{multline*}
  The sum on the last line vanishes, since
  $(-1)^{\binom{n+2}{2}} = - (-1)^{\binom{n}{2}}$. We conclude that
  \begin{multline*}
    \s \sum_{n=-\infty}^\infty (-1)^{\binom{n}{2}} \,
    \T^{\nu_0\dots\nu_k}( \Psi_{-n},\Psi_{n-k-2} ) \\
    = 2 \sum_{j=0}^k (-1)^{k-j} \eta^{\mu\nu_j}p_\mu
    \sum_{n=-\infty}^\infty
    (-1)^{\binom{n}{2}} \T^{\nu_0\dots\widehat{\nu}{}_j\dots\nu_k}(
    \Psi_{-n}, \Psi_{n-k-1} ) ,
%
%    &=
%      \begin{cases}
%        \displaystyle - 2 \sum_{n=-\infty}^\infty
%        (-1)^{\binom{n}{2}} \, \T(\Psi_{-n},\Psi_{n-1}) , & k=0 ,
%        \\[15pt]
%        \displaystyle - 2 \sum_{j=0}^k (-1)^{k-j}
%        \phi_{\nu_0\dots\widehat{\nu}{}_j\dots\nu_k} , & k>0 .
%      \end{cases}
%
  \end{multline*}
  from which \eqref{GGG0} follows.
\end{proof}

\begin{corollary}
  The long exact sequence
  \begin{equation*}
    \begin{tikzcd}
      \cdots \arrow{r}{\p} & H^{-1}(\AAhat,\s) \arrow{r} & H^{-1}(\FF,\s)
      \arrow{r} & H^0(\tAA,\s)
      \arrow[out=-5,in=170,overlay]{lld}[']{\p} & \\
      & H^0(\AAhat,\s) \arrow{r}{\p} & H^0(\FF,\s) \arrow{r} &
      H^1(\tAA,\s) \arrow{r}{\p} & \cdots
  \end{tikzcd}
\end{equation*}
splits, in the sense that the morphisms $\p$ vanish.
\end{corollary}

By an extension of this method, we may show that the space of
solutions of \eqref{covariant} is a contractible simplicial set. This
amounts to showing that for each $n>0$, any solution of
\eqref{covariant} in $\Om(\p\Delta^n)\otimes\|\FF(N_\bull\CU)[[u]]\|$
may be extended to a solution of \eqref{covariant} in
\begin{equation*}
  \Om(\Delta^n)\otimes\|\FF(N_\bull\CU)[[u]]\| =
  \Om_n\otimes\|\FF(N_\bull\CU)[[u]]\| .
\end{equation*}
In particular, the case $n=1$ shows that there is a solution of
\eqref{covariant} in $\Om_1\otimes\|\FF(N_\bull\CU)[[u]]\|$
interpolating between any pair of solutions of \eqref{covariant} in
$\|\FF(N_\bull\CU)[[u]]\|$.

\renewcommand{\c}{\mathsf{c}}
\newcommand{\x}{\mathsf{x}}

The cover $\CU$ of $M_0$ may be used to give explicit formulas for
cohomology classes in the hypercohomology of the complexes of sheaves
$\FF$ and $\AAhat$. The 1-cochain in the Thom-Whitney complex of
$\AAhat$
\begin{equation*}
  \c = c - \sum_{\alpha=0}^9 q_\alpha \left( p_\mu \T( \gamma^\mu
  \gamma^\alpha \theta_0,\theta_1 ) + 2 e^+
  \T^\alpha(\theta_1,\theta_1 ) -
  2 e^+ \T^\alpha(\theta_0,\theta_2 ) \right)
\end{equation*}
is a cocycle, and the 0-cochain
\begin{equation*}
  \x^\mu = x^\mu - \half \sum_{\alpha=0}^9 q_\alpha \left( p_\nu
    \T(\gamma^\alpha\gamma^\nu\gamma^\mu\theta_0,\theta_0) -
    4e^+ \T^{\alpha\mu}(\theta_0,\theta_1) \right)
\end{equation*}
satisfies the formula $(\delta+\s)\x^\mu=-\eta^{\mu\nu}\c p_\nu$,
analogous to the formula $sx^\mu=-\eta^{\mu\nu}cp_\nu$ for the
particle. (In the definitions of $\c$ and $\x^\mu$, we understand the
Einstein summation convention for the indices $\mu$ and $\nu$, but not
$\alpha$.)  We see that if $f(x,p,\theta)$ is a function of $x^\mu$,
$p_\mu$ and $\theta\in\Ss_-$, then
\begin{equation*}
  (\delta+\s)f(\x,p,\uptheta) = - \c p_\mu\,\frac{\p f}{\p
    x^\mu}(\x,p,\uptheta) ,
\end{equation*}
and $(\delta+\s)\c f(\x,p,\uptheta)=0$, where
$\uptheta=p_\mu\gamma^\mu\theta-2e^+\theta_1$. This may be compared to
those of Section~2 of Bergshoeff et al.\ \cite{BKV}; the formulas
presented there only apply outside the hypersurface $p_0-p_9=0$.

\section{Supersymmetry and Lorentz invariance of the solution}

The reason for the interest of the superparticle, and of the
Green--Schwarz superstring for which it is a toy model, is that it is
manifestly supersymmetric. The supersymmetry is generated by the
functional $\tint Q$, where
\begin{equation*}
  Q = \theta^+_0 - \half x^+_\mu \gamma^\mu \theta_0 \in
  \Ss_-\otimes\AA^{-1} .
\end{equation*}
The formula $\s Q = \p \bigl( p_\mu\gamma^\mu\theta_0 + 2e^+\theta_1)$
implies the vanishing of the Batalin--Vilkovisky antibracket
\begin{equation}
  \label{SQ}
  ( \tint Q , \tint \S ) = 0 .
\end{equation}

Let $\AA_\star$ be the subalgebra of the sheaf $\AA$ generated by the
fields
\begin{equation*}
  \{\p^\ell p_\mu,\p^\ell x^+_\mu,\p^\ell e^+,\p^\ell c^+\}_{\ell\ge0} \cup
  \{\p^\ell\Psi_n\mid n\in\Z\}_{\ell\ge0} .
\end{equation*}
Let $\AAhat_\star\subset\AAhat$ be its associated completion, with
respect to the fields of negative degree
\begin{equation*}
  \{\p^\ell x^+_\mu,\p^\ell e^+,\p^\ell c^+\}_{\ell\ge0} \cup
  \{\p^\ell\Psi_n\mid n<0\}_{\ell\ge0} .
\end{equation*}
Both $\AA_\star$ and $\AAhat_\star$ may be viewed as sheaves over the
fibre $\Phase_\star$ of $\Phase_0$ over the point $x^\mu=0$.

\begin{lemma}
  The subsheaf $\AAhat_\star\subset\AAhat$ satisfies
  $\p[\AAhat_\star]\subset\AAhat_\star$ and is closed under the
  Soloviev bracket.
\end{lemma}
\begin{proof}
  It follows directly from its definition that $\AAhat_\star$ is
  preserved by the action of $\p$. In order for $\AAhat_\star$ to be
  closed under the Soloviev bracket, it suffices to observe that for
  all fields $\Phi$ that generate $\AA_\star$, we have
  \begin{equation*}
    \frac{\p\Phi}{\p(\p^k x^\mu)} = \frac{\p\Phi}{\p(\p^k p^{+\mu})} =
    \frac{\p\Phi}{\p(\p^k e^+)} = \frac{\p\Phi}{\p(\p^k c)} = 0 .
  \end{equation*}
  This implies that
  \begin{multline*}
    \[ f , g \] = \sum_{n=0}^\infty (-1)^{(n+1)(\pa(f)+1)}
    \sum_{k,\ell=0}^\infty \\
    \left( \p^\ell \left( \frac{\p
          f}{\p(\p^k\theta_n)} \right) \, \p^k \left( \frac{\p
          g}{\p(\p^\ell\theta^+_n)} \right) + (-1)^{\pa(f)} \p^\ell
      \left( \frac{\p f}{\p(\p^k\theta^+_n)} \right) \, \p^k \left(
        \frac{\p g}{\p(\p^\ell\theta_n)} \right) \right) .
  \end{multline*}
  It only remains to observe that for all $m\in\Z$ and $n\ge0$, and
  all $k,\ell\ge0$, the partial derivatives
  $\p(\p^\ell\Psi_m)/\p(\p^k\theta_n)$ and
  $\p(\p^\ell\Psi_m)/\p(\p^k\theta^+_n)$ are in $\AAhat_\star$.
\end{proof}

We now have the following analogue of Theorem
\ref{superparticle:main}. The proof follows the same lines, but is
actually somewhat simpler.

\begin{lemma}
  \label{superparticle:subalgebra}
  Let $\FF_\star=\AAhat_\star/\p\AAhat_\star$. The cohomology sheaf
  $H^i(\FF_\star,\s)$ vanishes unless $i\in\{-1,0\}$.
\end{lemma}
\begin{proof}
  The sheaf $\AAhat_\star$ is an algebra over the momentum space
  $\Phase_\star$, whose structure sheaf is the algebra of rational
  functions in the variables $\{p_\mu\}$. The filtration of $\AAhat$
  induces a filtration of $\AAhat_\star$, and the differential $\s_0$
  on the zeroth page of the associated spectral sequence $E^{pq}_0$
  equals
  \begin{equation*}
    \s_0 = - \pr \left( \p p_\mu \, \frac{\p~}{\p x^+_\mu}
      + \p e^+ \, \frac{\p~}{\p c^+} \right) .
  \end{equation*}
  This is a Koszul differential and its cohomology $E_1$ is the graded
  commutative algebra freely generated over the structure sheaf
  $\CO_{\Phase_\star}$ by the variables
  \begin{equation*}
    \{ e^+ \} \cup \{ \p^\ell\Psi_n \mid n\in\Z \}_{\ell\ge0} .
  \end{equation*}
  
  The differential $\s_1$ on the first page $E_1$ of the spectral
  sequence is given by the formula
  \begin{equation*}
    \s_1 = \pr \left( - \tfrac{1}{2} \eta^{\mu\nu} p_\mu p_\nu \,
      \frac{\p~}{\p e^+} \right) .
  \end{equation*}
  The element $\eta^{\mu\nu}p_\mu p_\nu$ is not a zero divisor in
  $E_1$: its zero-locus is the light-cone
  \begin{equation*}
    \Cone_\star = \{ p_\mu \ne 0 \mid \eta^{\mu\nu} p_\mu p_\nu = 0 \} .
  \end{equation*}

  We see that the second page $E_2$ of the spectral sequence is a
  sheaf of graded commutative algebras generated over
  $\CO_{\Cone_\star}$ by the variables
  \begin{equation*}
    \{ \p^\ell\Psi_n \mid n\in\Z \}_{\ell\ge0} .
  \end{equation*}
  The differential $\s_2$ on the second page $E_2$ of the spectral
  sequence is given by the formula
  \begin{equation*}
    \s_2 = \sum_{n=1}^\infty (-1)^{n+1} \pr \left( p_\mu T^\mu\left(
        \Psi_{1-n} , \frac{\p~}{\p\Psi_{-n}} \right) \right) .
  \end{equation*}
  On the light-cone, the operator
  \begin{equation*}
    p_\mu \gamma^\mu : \Ss_\pm \to \Ss_\mp
  \end{equation*}
  has vanishing cohomology. We see that the third page $E^{p+q}_3$ is
  generated over $\CO_{\Cone_\star}$ by the variables
  \begin{equation*}
    \{\p^\ell\Psi_n\mid n\ge0\}_{\ell\ge0}\in E^{3n,-2n}_3 ,
  \end{equation*}
  modulo relations
  $\{p_\mu\gamma^\mu\p^\ell\Psi_0\}_{\ell\ge0}\in E^{00}_3$ and hence
  the differential $\s_r$ of the $r$th page of the spectral sequence
  vanishes for $r>3$.

  The remainder of the proof follows the proof of
  Theorem~\ref{superparticle:main}.
\end{proof}

\begin{theorem}
  \label{supersymmetry}
  There is a choice of the solution $\S_u$ to the equation
  \eqref{covariant} such that
  \begin{equation*}
    ( \tint Q , \tint \S_u ) = 0 .
  \end{equation*}
\end{theorem}
\begin{proof}
  Let $\q$ be the Hamiltonian vector field associated to $\tint Q$. It
  is easily seen that $\q\Psi_n=0$, and hence that $\q$ annihilates
  $\AAhat_\star$. It is easily seen that $\q\G_0=0$. We prove the
  theorem by showing that for all $n>0$, $\G_n$ may be chosen in
  $\|\FF_\star(N_\bull\CU)\|^{-2n-2}$. In view of
  Lemma~\ref{superparticle:subalgebra}, it suffices to show that the
  cocycle
  \begin{equation*}
    - \frac{1}{2} \sum_{j+k=n-1} ( \tint \G_j , \tint \G_k )
  \end{equation*}
  lies in $\|\FF_\star(N_\bull\CU)\|^{-2n-1}$ for $n>0$. By induction, we
  may assume that this holds for all of the terms of this sum with
  $j,k>0$. It remains to check that
  \begin{equation*}
    ( \tint \G_0 , \tint \G_{n-1} )  \in \|\FF_\star(N_\bull\CU)\|^{-2n-1} .
  \end{equation*}
  But modulo $\AAhat_\star$, $\G_0=x^+_\mu p^{+\mu}+ec^+$, and it is
  easily seen that
  \begin{equation*}
    \[ x^+_\mu p^{+\mu}+ec^+ , \AAhat_\star \] \subset \AAhat_\star .
  \end{equation*}
  Indeed, on restriction to $\AAhat_\star$, the Soloviev bracket
  $\ad(x^+_\mu p^{+\mu}+ec^+)$ is given by the evolutionary vector field
  \begin{equation*}
    \pr \left( - x^+_\mu \, \frac{\p~}{\p p_\mu} + c^+ \,
      \frac{\p~}{\p e^+} \right) ,
  \end{equation*}
  which preserves $\AAhat_\star$.
\end{proof}

We now turn to the question of Lorentz invariance of the solution
$\S_u$ to \eqref{covariant} that we have obtained. Let $\so(9,1)$ be
the Lie algebra of the Lorentz group $\SO(9,1)$, with basis
$\rho^{\mu\nu}=-\rho^{\nu\mu}$, $0\le\mu<\nu\le9$, satisfying the
commutation relations
\begin{equation*}
  [ \rho^{\kappa\lambda} , \rho^{\mu\nu} ] = \eta^{\lambda\mu}
  \rho^{\kappa\nu} + \eta^{\kappa\nu} \rho^{\lambda\mu} -
  \eta^{\lambda\nu} \rho^{\kappa\mu} - \eta^{\kappa\mu}
  \rho^{\lambda\nu} .
\end{equation*}
This action is realized on $\FF$ by the currents of \eqref{Lorentz_alg}.

Consider the complex
\begin{equation*}
  C^*(\so(9,1)) \otimes \|\FF(N_\bull\CU)\| .
\end{equation*}
The antibracket on $\|\FF(N_\bull\CU)\|$ induces a graded Lie bracket
on $C^*(\so(9,1)) \otimes \|\FF(N_\bull\CU)\|$, making it into a
differential graded Lie algebra. Denote the dual basis of
$\so(9,1)^\vee$ by $\eps_{\mu\nu}$, and consider the element
\begin{equation*}
  \S(\eps) = \S + M^{\mu\nu} \eps_{\mu\nu} \in C^*(\so(9,1)) \otimes
  \FF(\Phase_0) \subset C^*(\so(9,1)) \otimes \|\FF(N_\bull\CU)\| .
\end{equation*}
The invariance of $\S$ under the action of the Lorentz group implies
that this is a Maurer--Cartan element of
$C^*(\so(9,1)) \otimes \|\FF(N_\bull\CU)\|$; see \eqref{Lorentz}.

We will construct a sequence of elements
\begin{equation*}
  \G_n(\eps) \in C^*(\so(9,1)) \otimes \|\FF(N_\bull\CU)\| ,
\end{equation*}
of total degree $-2n-2$, supersymmetric $\q\G_n(\eps)=0$, such that
the series
\begin{equation*}
  \S_u(\eps) = \S(\eps) + \sum_{n=0}^\infty u^{n+1} \G_n(\eps) \in
  C^*(\so(9,1)) \otimes \|\FF(N_\bull\CU)\|[[u]]
\end{equation*}
satisfies the equation \eqref{covariant:Lorentz}.

We solve \eqref{covariant:Lorentz} inductively, by an extension of the
method used to prove Theorem~\ref{supersymmetry}. Write
\begin{equation*}
  \G_n(\eps) = \sum_{k=0}^{10} \G_{n,k} ,
\end{equation*}
where $\G_{0,0}$ equals the explicit solution
$\G_0\in \|\FF(N_\bull\CU)\|^{-2}$ of \eqref{G0:explicit}, and
$\G_{n,k}\in C^k(\so(9,1)) \otimes
\|\FF_\star(N_\bull\CU)\|^{-2n-k-2}$ for $n>0$ or $k>0$. Assuming we
have found $\G_{m,\ell}$ for $m<n$ or $m=n$ and $\ell<k$, we must
solve the equation
\begin{multline}
  \label{Gnk}
  ( \delta + \s ) \G_{n,k} = - d\G_{n,k-1} - ( M^{\mu\nu}
  \eps_{\mu\nu} , \G_{n,k-1} ) \\
  - \frac{1}{2} \sum_{m=0}^{n-1} \sum_{\ell=0}^k ( \G_{m,\ell} ,
  \G_{n-m-1,k-\ell} ) \in C^k( \so(9,1) ) \otimes \|
  \FF_\star(N_\bull\CU) \|^{-2n-k-1} .
\end{multline}
By the Lorentz invariance of $\S$, $\s M^{\mu\nu}=0$. For $n>0$ or
$k>0$, this is sufficient to imply that the right-hand side of
\eqref{Gnk} is a cocycle. In the case $n=0$ and $k=1$, we need in
addition the formula
\begin{equation*}
  ( M^{\mu\nu} , \D ) = 0 .
\end{equation*}
By Lemma~\ref{superparticle:subalgebra}, there is a solution
\begin{equation*}
  \G_{n,k}\in C^k(\so(9,1)) \otimes \|\FF_\star(N_\bull\CU)\|^{-2n-k-2} .
\end{equation*}
Thus there exists a supersymmetric solution to the equation
\eqref{covariant:Lorentz}.

\appendix

\section{Spinors in signature \texorpdfstring{$(9,1)$}{(9,1)}}

Let $\Spin(9,1)$ be the double cover of the proper Lorentz group
$\SO_+(9,1)$. Let $\Ss_+$ and $\Ss_-$ be the left and right-handed
Majorana--Weyl spinor representations: these are $16$-dimensional real
representations of $\Spin(9,1)$. Let us review their construction.

Let $W$ be a finite-dimensional oriented real vector space with a
non-degenerate inner product $(v,w)$ of dimension $n$. Let
$\mathbb{T}:\Wedge^nW\to\R$ be the linear map which takes the wedge
product of an oriented unimodular frame of $W$ to $1$. This induces an
inner product on the exterior algebra $\Wedge^*W$, defined on
$\alpha\in\Wedge^kW$ and $\beta\in\Wedge^\ell W$ by the formula
\begin{equation*}
  \<\alpha,\beta\> = (-1)^{\binom{k}{2}} \mathbb{T} \bigl( \alpha
  \wedge \beta \bigr) .
\end{equation*}
Thus $\<\alpha,\beta\>$ vanishes unless $k+\ell=n$, and the form
satisfies the symmetry
\begin{equation*}
  \<\beta,\alpha\> = (-1)^{\binom{n}{2}} \<\alpha,\beta\> .
\end{equation*}

The endomorphism algebra of the exterior algebra $\Wedge^*W$ is the
Clifford algebra generated by the operators
$c^\pm(v)=\iota(v)\pm\eps(v)$, where
$\eps(v):\Wedge^*W\to\Wedge^{*+1}W$ is exterior multiplication by
$v\in W$ and $\iota(v):\Wedge^{*-1}W\to\Wedge^*W$ is contraction with
$v$. It is easily seen that $\eps(v)$ and $\iota(v)$, and hence
$c^\pm(v)$, are self-adjoint for the bilinear form $\<\alpha,\beta\>$.

The endomorphism algebra of the real vector space underlying the
quaternions $\H$ is the tensor product of the two commuting
subalgebras of left and right multiplication in $\H$. (This is one way
of seeing the isomorphism $\Spin(4)\cong\SU(2)\times\SU(2)$.) Denote
left, respectively right, multiplication by an element $a\in\H$ by
$a_L$, respectively $a_R$.

We now give an explicit representation of the Clifford algebra in
signature $(9,1)$ acting on the space of spinors
$\Ss=\H\otimes\Wedge^*\R^3$:
\begin{align*}
  \gamma^1
  &= c^+_1
  &
  \gamma^2
  &= i_L c^-_1
  &
  \gamma^3
  &= j_L c^-_1
  &
  \gamma^4
  &= k_L c^-_1 \\
  \gamma^5
  &= c^+_2
  &
  \gamma^6
  &= i_R c^-_2
  &
  \gamma^7
  &= j_R c^-_2
  &
  \gamma^8
  &= k_R c^-_2 \\
  \gamma^9
  &= c^+_3
  &
  \gamma^0
  &= c^-_3
\end{align*}
The $\gamma$-matrices $\gamma^\mu$ exchange the subspaces
$\Ss_\pm\cong\R^{16}=\H\otimes\Wedge^\pm\R^3$ of
$\H\otimes\Wedge^*\R^3$ of even, respectively odd, exterior degree in
the exterior algebra.

The Lie algebra of the group $\Spin(9,1)$ is spanned by the quadratic
expressions in the $\gamma$-matrices
\begin{equation*}
  \gamma^{\mu\nu} = \half \bigl( \gamma^\mu\gamma^\nu -
  \gamma^\nu\gamma^\mu \bigr) .
\end{equation*}
In particular, $\Spin(9,1)$ preserves the subspaces $\Ss_\pm$ of
$\Ss$. These are the left and right-handed $16$-dimensional
Majorana--Weyl spinors representations of $\Spin(9,1)$.

There is a non-degenerate symmetric bilinear form on $\Ss$, given by
the formula
\begin{equation*}
  \T( \alpha , \beta ) = \mathbb{T} \Re\bigl( c^-_1c^-_2\alpha
  \wedge \overline{\beta} \bigr) : \Ss \otimes \Ss \to \R .
\end{equation*}
From the explicit formulas for $\gamma^\mu$, we see that
\begin{equation*}
  \T( \gamma^\mu\alpha , \beta ) = \T( \alpha , \gamma^\mu
  \beta ) .
\end{equation*}
The Clifford algebra of $\R^{9,1}$ has basis
\begin{equation*}
  \gamma^{\mu_1\dots\mu_k} = \frac{1}{k!} \sum_{\pi\in S_k} (-1)^\pi
  \gamma^{\mu_{\pi(1)}} \dots \gamma^{\mu_{\pi(k)}} ,
\end{equation*}
where $\mu_1\dots\mu_k$ ranges over the set
$\{1\le \mu_1<\cdots<\mu_k\le 10\}$. Let $\CV$ be the vector
representation $\R^{9,1}$ of $\Spin(9,1)$, and define pairings
$\T^{\mu_1\dots\mu_k} : \Ss \otimes \Ss \to \Wedge^k \CV$ by
\begin{equation*}
  \T^{\mu_1\dots\mu_k}(\alpha,\beta) =
  \T(\gamma^{\mu_1\dots\mu_k}\alpha,\beta) =
  (-1)^{\binom{k}{2}} \, \T(\alpha,\gamma^{\mu_1\dots\mu_k}\beta) .
\end{equation*}

\begin{lemma}
  \label{commute}
  \begin{equation*}
    \gamma^{\mu_1\dots\mu_k}\gamma^\mu - (-1)^k \,
    \gamma^\mu \gamma^{\mu_1\dots\mu_k} = 2 \sum_{j=1}^k (-1)^{k-j}
    \eta^{\mu\mu_j} \gamma^{\mu_1\dots\widehat{\mu}{}_j\dots\mu_k}
  \end{equation*}
\end{lemma}

\section*{Acknowledgments}

{\small

  The first author is grateful to Chris Hull for introducing him to
  the superparticle. His research is partially supported by a
  Fellowship of the Simons Foundation, Collaboration Grants 243025 and
  524522 of the Simons Foundation, and EPSRC Programme Grant
  EP/K034456/1 ``New Geometric Structures from String Theory.'' Parts
  of this paper were written while he was visiting Imperial College,
  the Yau Mathematical Sciences Center at Tsinghua University and the
  Department of Mathematics of Columbia University, as a guest of
  Chris Hull, Si Li and Mohammed Abouzaid respectively.

  The research of the second author is supported in part by the
  National Science Foundation grant ``RTG: Analysis on manifolds'' at
  Northwestern University.

}

\address{Northwestern University, Evanston, Illinois, USA}

\end{document}